\documentclass[11pt,a4paper]{article}
\usepackage{amssymb,amsmath,amsthm}
\usepackage[english]{babel}
\usepackage{a4wide}
\usepackage{microtype}
\usepackage{hyperref}
\hypersetup{colorlinks=true,citecolor=blue,linkcolor=blue,urlcolor=blue}
\usepackage{bbm}
\usepackage{subcaption}
\usepackage{enumerate}
\usepackage[ruled,linesnumbered,lined,boxed]{algorithm2e}

\newtheorem{theorem}{Theorem}

\newtheorem{proposition}[theorem]{Proposition} 
\newtheorem{observation}[theorem]{Observation} 
\theoremstyle{definition} 
\newtheorem{remark}[theorem]{Remark} 
\newtheorem{definition}[theorem]{Definition}

\newtheorem*{proposition*}{Proposition} 
\newtheorem*{theorem*}{Theorem} 

\newcommand{\A}{\mathbf{A}}

\newcommand{\X}{\mathbf{X}}
\newcommand{\NAE}{\mathbf{NAE}}
\newcommand{\B}{\mathbf{B}}
\newcommand{\C}{\mathbf{C}}
\newcommand{\Ss}{\mathbf{S}}
\newcommand{\T}{\mathbf{T}}
\newcommand{\AT}{\operatorname{AT}}
\newcommand{\XOR}{\operatorname{XOR}}
\newcommand{\AND}{\operatorname{AND}}
\newcommand{\OR}{\operatorname{OR}}
\newcommand{\MAJ}{\operatorname{MAJ}}
\newcommand{\inn}[2]{\mathbf{#1}\textbf{-in-}\mathbf{#2}}
\newcommand{\odd}[1]{\mathbf{odd}\textbf{-in-}\mathbf{#1}}
\newcommand{\THRo}[1]{\operatorname{THR}_{#1}}
\newcommand{\THR}[2]{\operatorname{THR}_{\frac{#1}{#2}}}

\newcommand{\Z}{\mathbb{Z}}

\DeclareMathOperator{\ar}{ar}
\DeclareMathOperator{\Pol}{Pol}
\DeclareMathOperator{\CSP}{CSP}
\DeclareMathOperator{\PCSP}{PCSP}
\DeclareMathOperator{\BLP}{BLP}
\DeclareMathOperator{\AIP}{AIP}

\renewcommand{\vec}[1]{\mathbf{#1}}
\newcommand{\bx}{\vec{x}}
\newcommand{\ba}{\vec{a}}
\newcommand{\by}{\vec{y}}
\newcommand{\red}{\leq_p}

\begin{document}

\author{Alex Brandts\\
University of Oxford\\
\texttt{alex.brandts@cs.ox.ac.uk}
\and
Stanislav \v{Z}ivn\'y\\
University of Oxford\\
\texttt{standa.zivny@cs.ox.ac.uk}
}

\title{Beyond $\PCSP(\inn{1}{3}, \NAE)$\thanks{An extended abstract of this work
appeared in the \emph{Proceedings of the 48th International Colloquium on
Automata, Languages, and Programming (ICALP'21)}~\cite{bz21:icalp}. 
Alex Brandts was supported by a Royal Society Enhancement Award and an NSERC PGS Doctoral Award. 
Stanislav \v{Z}ivn\'y was supported by a Royal Society University Research
Fellowship. This project has received funding from the European Research Council
(ERC) under the European Union's Horizon 2020 research and innovation programme
(grant agreement No 714532). The paper reflects only the authors' views and not
the views of the ERC or the European Commission. The European Union is not
liable for any use that may be made of the information contained therein.}}

\maketitle

\begin{abstract}

  The promise constraint satisfaction problem (PCSP) is a recently introduced
  vast generalisation of the constraint satisfaction problem (CSP) that captures
  approximability of satisfiable instances. A PCSP instance comes with two forms
  of each constraint: a strict one and a weak one. Given the promise that a
  solution exists using the strict constraints, the task is to find a solution
  using the weak constraints. While there are by now several dichotomy results
  for fragments of PCSPs, they all consider (in some way) symmetric PCSPs.

  1-in-3-SAT and Not-All-Equal-3-SAT are classic examples of Boolean symmetric
  (non-promise) CSPs. While both problems are NP-hard, Brakensiek and Guruswami
  showed~[SICOMP'21] that given a satisfiable instance of 1-in-3-SAT one can find
  a solution to the corresponding instance of (weaker) Not-All-Equal-3-SAT. In
  other words, the PCSP template $(\inn{1}{3},\NAE)$ is tractable. 

  We focus on \emph{non-symmetric} PCSPs. In particular, we study PCSP templates
  obtained from the Boolean template $(\inn{t}{k}, \NAE)$ by either adding
  tuples to $\inn{t}{k}$ or removing tuples from $\NAE$. For the former, we
  classify all templates as either tractable or not solvable by one of 
  the strongest known algorithm for PCSPs, the combined basic LP and affine IP
  relaxation of Brakensiek, Guruswami, Wrochna, and \v{Z}ivn\'y~[SICOMP'20]. For the latter, we classify
  all templates as either tractable or NP-hard.

\end{abstract}

\section{Introduction}
\label{sec:intro}

How hard is it to find a 6-colouring of a graph if it is promised to be
3-colourable? We do not know but believe it to be NP-hard. Despite sustained
effort, this so-called \emph{approximate graph colouring} problem has been
elusive since it was considered by Garey and Johnson almost 50 years
ago~\cite{GJ76}. The current state of the art is
NP-hardness of finding a 5-colouring of a 3-colourable graph~\cite{BBKO21}.
Approximate graph colouring is an example of the very general promise constraint
satisfaction problem, which is the focus of this paper. We start with
(non-promise) constraint satisfaction problems to set the stage.

\paragraph{Constraint satisfaction}
While deciding whether a graph is 2-colourable is solvable in polynomial
time, deciding 3-colourability is NP-complete~\cite{Karp72:reducibility}.
The \emph{constraint satisfaction problem} (CSP) is a general framework that
captures graph colourings and many other fundamental computational problems. Feder and Vardi
initiated a systematic study of so-called fixed-template decision
CSPs. Let $\A$ be a fixed finite relational structure, called the
\emph{template} or constraint language; i.e., $\A$ consists of a finite universe
$A$ and finitely many relations on $A$, each of possibly different arity. The
fixed-template CSP over $\A$, denoted by $\CSP(\A)$, is the class of CSPs in
which all constraint relations come from $\A$. In more detail, $\CSP(\A)$ denotes
the following computational problem: Given a structure $\X$ over the same
signature as $\A$, 
is there a homomorphism from $\X$ to $\A$, denoted
by $\X\to\A$?  
(Formal definitions can be found in Section~\ref{sec:prelims}.)
If $\A=K_3$ is a clique on 3 vertices then $\CSP(\A)$ is
precisely the standard graph 3-colouring problem.

A classic result of Schaefer shows that, for any $\A$ on a 2-element set,
$\CSP(\A)$ is either solvable in polynomial time or NP-complete. The non-trivial
tractable cases from Schaefer's classification are taught in undergraduate
algorithms courses: 2-SAT, (dual) Horn-SAT, and linear equations over $\{0,1\}$.
Two concrete CSPs that are NP-hard by Schaefer's result are the (positive)
1-in-3-SAT and (positive) Not-All-Equal-3-SAT. For both problems, the instance
is a list of triples of variables. In 1-in-3-SAT, the task is to find a mapping
from the variables to $\{0,1\}$ so that in each specified triple exactly one variable is
set to $1$. Formally, 1-in-3-SAT is $\CSP(\inn{1}{3})$, where
$\inn{1}{3}=(\{0,1\};\{(1,0,0),(0,1,0),(0,0,1)\})$. In Not-All-Equal-3-SAT, the task is
to find a mapping from the variables to $\{0,1\}$ so that in each triple not all
variables are assigned the same value. Formally, Not-All-Equal-3-SAT is
$\CSP(\NAE)$, where $\NAE=(\{0,1\};\{0,1\}^3\setminus\{(0,0,0),(1,1,1)\})$.

If $\A$ is a graph (i.e., a single symmetric binary relation) then, as shown by
Hell and Ne\v{s}et\v{r}il~\cite{Hell90:h-coloring}, $\CSP(\A)$ is either
solvable in polynomial time or NP-complete. 

Based on these two examples and a connection to logic, Feder and Vardi famously
conjectured~\cite{Feder98:monotone} that, for any finite $\A$, $\CSP(\A)$ is
either solvable in polynomial time or NP-complete.
Bulatov~\cite{Bulatov17:focs}, and independently Zhuk~\cite{Zhuk20}, proved the
conjecture in the affirmative, both relying on the algebraic approach to
CSPs~\cite{Jeavons97:jacm,Bulatov05:classifying,BKW17}. 

\paragraph{Promise constraint satisfaction} 
Austrin, Guruswami, and H{\aa}stad~\cite{AGH17} and Brakensiek and
Guruswami~\cite{BG21} initiated the investigation of the
\emph{promise constraint satisfaction problem} (PCSP), which is a vast
generalisation of the CSP. Let $\A$ and $\B$ be two relational structures such
that $\A\to\B$. The fixed-template PCSP over $\A$ and $\B$, denoted by
$\PCSP(\A,\B)$, is the following computational problem: Given $\X$ such that
$\X\to\A$, find a homomorphism from $\X$ to $\B$ (which exists by the
composition of the promised homomorphism from $\X$ to $\A$ and the homomorphism
from $\A$ to $\B$). If we take $\A=K_3$ to be a clique on 3
vertices and $\B=K_6$ to be a clique on 6 vertices, then $\PCSP(\A,\B)$ is
an instance of the approximate graph colouring problem mentioned at
the beginning of this article. 

Actually, what we described is the \emph{search} version of the PCSP. The
\emph{decision} version is as follows: Given $\X$, return \textsc{Yes} if
$\X\to\A$ and return \textsc{No} if $\X\not\to\B$. (The promise in the decision
version is that it does not happen that $\X\not\to\A$ but $\X\to\B$.) It is well
known that the decision version reduces to the search version but it is not
known whether there is a reduction the other way~\cite{BBKO21}. In most results
(including ours), hardness is established for the decision version and
tractability for the search version.

If $\A=\B$ then $\PCSP(\A,\B)$ is the same as $\CSP(\A)$ and
thus PCSPs indeed generalise CSPs. For CSPs, the decision and search versions are
known to be equivalent~\cite{Bulatov05:classifying}.

Building on the result of Barto, Opr\v{s}al, and Pinsker~\cite{BOP18} that the
complexity of $\CSP(\A)$ is captured by certain types of identities of
higher-order symmetries (called polymorphisms) of $\A$, Barto, Bul\'in, Krokhin,
and Opr\v{s}al showed that the basics of the algebraic approach developed for
CSPs~\cite{BOP18} can be generalised to PCSPs~\cite{BBKO21}, thus
introducing a general methodology for investigating the computational complexity
of PCSPs. In particular, among other things, they showed that finding a
5-colouring of a 3-colourable graph is NP-hard.

\paragraph{Related work}
Motivated by the goal to understand the computational complexity of all
fixed-template PCSPs, a recent line of research has focused on restricted
classes of templates, with the main directions being Boolean templates (i.e.,
templates on a two-element set) and symmetric templates (i.e., all relations in
the template satisfy that if a tuple belongs to a relation then so do all
its permutations).

Austrin, Guruswami, and H{\aa}stad~\cite{AGH17} considered the $(1,g,k)$-SAT
problem: Given an instance of $k$-SAT with the promise that there is an
assignment satisfying at least $g$ literals in each clause, find an assignment
that satisfies at least one literal in each clause. They showed that this
problem is NP-hard if $\frac{g}{k}<\frac{1}{2}$, and polynomial-time solvable
otherwise. $(1,g,k)$-SAT is a Boolean PCSP with a (symmetric) template that
includes the binary disequality relation and a relation containing all tuples of particular Hamming weights. The NP-hardness in~\cite{AGH17} was proved via
reduction from the label cover problem using the idea of polymorphisms lifted
from CSPs to PCSPs. Building on the algebraic theory from~\cite{BBKO21},
Brandts, Wrochna, and \v{Z}ivn\'y~\cite{BWZ21} extended the classification of
$(1,g,k)$-SAT to arbitrary finite domains.

Brakensiek and Guruswami~\cite{BG21} managed to classify all PCSPs over
symmetric Boolean templates with the disequality relation as NP-hard or
solvable in polynomial time. Ficak, Kozik, Ol\v{s}\'ak, and
Stankiewicz~\cite{fkos19} extended this result to all symmetric Boolean
templates.

In very recent work, Barto, Battistelli, and Berg~\cite{Barto21:stacs}
explored symmetric PCSPs on three- and four-element domains.

While the approximate graph colouring problem remains open, hardness was proved
under stronger assumptions (namely Khot's 2-to-1 Conjecture~\cite{Khot02stoc}
for $k$-colourings with $k\geq 4$ and its non-standard variant for 3-colourings)
by Dinur, Mossel, and Regev~\cite{DMR09}. Guruswami and
Sandeep~\cite{GS20:icalp} recently established this result under a weaker
assumption, the so-called $d$-to-1 conjecture for any fixed $d\geq 2$. For
approximate \emph{hypergraph} colouring, another important PCSP, NP-hardness
was established by Dinur, Regev, and Smyth~\cite{DRS05}. There has been some
recent progress on approximate graph colourings~\cite{WZ20} and related PCSPs,
e.g. approximate graph homomorphism problems~\cite{KO19,WZ20}, and rainbow vs.
normal hypergraph colourings~\cite{ABP20}. 

\paragraph{Contributions}
Unlike most previous works, which focused on symmetric PCSPs, we investigate
\emph{non-symmetric} PCSPs. Our first motivation is that a classification of more
concrete PCSP templates is needed to improve and extend the general
algebraic theory from~\cite{BBKO21},
for example by identifying new hardness and tractability criteria. 
At the moment, even an analogue of Schaefer's result, i.e., classifying
all Boolean PCSPs, seems out of reach. Our second motivation is the pure beauty
of the template $(\inn{1}{3}, \NAE)$. While $\PCSP(\inn{1}{3}, \NAE)$ admits a
polynomial-time algorithm~\cite{BG21,BG19}, tractability cannot be
obtained via a ``gadget reduction'' to tractable finite-domain
CSPs~\cite{BBKO21} or via a ``local consistency checking''~\cite{Atserias22:soda}.

Let $\inn{t}{k}$ denote the
Boolean structure with a single relation of arity $k$ that contains tuples with
exactly $t$ 1's and let $\NAE$ denote the Boolean structure with a single relation
of arity $k$, which is always clear from the context, that contains all tuples
except for the two all-equal tuples. (Previously in this section, we used $\NAE$
only with $k=3$.)
Consider the Boolean $\PCSP(\inn{t}{k}, \NAE)$, which is a natural
generalisation of $\PCSP(\inn{1}{3}, \NAE)$. Similarly to
$\PCSP(\inn{1}{3},\NAE)$, we have that $\PCSP(\inn{t}{k},\NAE)$ is a symmetric
tractable PCSP. 

We study the following two questions: Firstly, when can we add
tuples to $\inn{t}{k}$ (i.e., how can we weaken the promise) to keep the PCSP
tractable? Secondly, when can we remove tuples from $\NAE$ (i.e., how can we
strengthen the relation $\NAE)$ to keep the PCSP tractable? Note that both of
these changes generally do not result in symmetric templates.

For the second question, we give a complete answer in Theorem~\ref{thm:rem}: If
$t$ is odd, $k$ is even, and tuples of only even Hamming weight are removed from
$\NAE$, the resulting PCSP is solvable in polynomial time. In all other cases,
the resulting PCSP is NP-hard. Put differently, $\PCSP(\inn{t}{k},\T)$ is
tractable if only if $\T=\NAE$ or $\CSP(\T)$ is tractable (assuming P$\neq$NP).

For the first question, we give a second-best possible answer in
Theorem~\ref{thm:add}: If $t$ is odd, $k$ is even, and tuples of only
odd Hamming weight are added to $\inn{t}{k}$, the resulting PCSP is tractable. In
all other cases, the resulting PCSP is not solved by the combined basic LP and
affine IP relaxation ($\BLP+\AIP$) of Brakensiek, Guruswami, Wrochna, and
\v{Z}ivn\'y~\cite{BGWZ20}, one of the currently
strongest known algorithm for PCSPs. The power of this relaxation, both in terms
of minions and polymorphism identities, is
known~\cite{BGWZ20}. It is consistent with the current knowledge that
$\BLP+\AIP$
could solve all tractable Boolean PCSPs.
The only stronger algorithm than $\BLP+\AIP$ studied in
the context of PCSPs is
CLAP~\cite{Ciardo22:soda}, but its power is currently only known via a
minion-theoretic characterisation (and not via a
polymorphism characterisation).
Similarly to $\PCSP(\inn{1}{3},\NAE)$, the PCSPs that we prove to be $\BLP+\AIP$-hard are not
solvable by ``local consistency checking'' and via a ``gadget reduction'' to
tractable finite-domain CSPs (cf. Remark~\ref{rm:intract}).

\medskip
One take-away message from our results is that the tractability of
$\PCSP(\inn{t}{k},\NAE)$ is very fragile, which gives more support for its
importance. Another message is that the PCSP templates obtained from
the template $(\inn{t}{k},\NAE)$ by adding a single tuple are good candidates for testing and/or
improving NP-hardness criteria for PCSPs. Finally,
Proposition~\ref{prop:firstsym}, while with a very simple proof, shows that the
classification of Boolean symmetric PCSP templates $(\A,\B)$ from~\cite{fkos19} holds more
generally and requires that only $\A$ should be symmetric.

\section{Preliminaries}
\label{sec:prelims}

We denote by $[n]$ the set $\{1,2,\ldots,n\}$. For a $k$-tuple $\bx$, 
we write $\bx=(x_1,\ldots,x_k)$. We denote by $\red$ a polynomial-time many-one
reduction and by $\equiv_p$ a polynomial-time many-one equivalence.

A \emph{relational structure} is a tuple $\A=(A;R_1,\ldots,R_p)$, where $A$ is a
finite set called the \emph{domain} of $\A$, and each $R_i$ is a relation of
arity $\ar(R_i)\geq 1$, that is, $R_i$ is a non-empty subset of $A^{\ar(R_i)}$.
A relational structure is \emph{symmetric} if each relation in it is invariant
under any permutation of coordinates. Two relational structures
$\A=(A;R_1,\ldots,R_p)$ and $\B=(B;S_1,\ldots,S_q)$ have the same
\emph{signature} if $p=q$ and $\ar(R_i)=\ar(S_i)$ for every $i\in[p]$. In this
case, a mapping $\phi:A\to B$ is called a \emph{homomorphism} from $\A$ to $\B$,
denoted by $\phi:\A\to\B$, if $\phi$ preserves all relations; that is, for every
$i\in[p]$ and every tuple $\bx\in R_i$, we have $\phi(\bx)\in S_i$, where $\phi$ is
applied component-wise. The existence of a homomorphism from $\A$ to $\B$ is
denoted by $\A\to\B$. A PCSP \emph{template} is a pair $(\A,\B)$ of relational
structures over the same signature such that $\A\to\B$.

\begin{definition}
  Let $(\A,\B)$ be a PCSP template. 
  The \emph{decision version} of $\PCSP(\A,\B)$ is the following problem: Given
  as input a relational structure $\X$ over the same signature as $\A$ and $\B$,
  output \textsc{Yes} if $\X\to\A$ and $\textsc{No}$ if $\X\not\to\B$.
  The \emph{search version} of $\PCSP(\A,\B)$ is the following problem: Given as
  input a relational structure $\X$ over the same signature as $\A$ and $\B$ such
  that $\X\to\A$, find a homomorphism from $\X$ to $\B$.
\end{definition}

We call $\PCSP(\A,\B)$ \emph{tractable} if any instance of $\PCSP(\A,\B)$ can be
solved in polynomial time in the size of the input structure $\X$.
It is easy to show that the decision version reduces to the search
version~\cite{BBKO21}. Our hardness results will be for the decision version and
our tractability results for the search version. For a relational
structure $\A$, the constraint satisfaction problem with the template $\A$,
denoted by $\CSP(\A)$, is $\PCSP(\A,\A)$.

The following notion of polymorphisms is at the heart of the algebraic approach
to (P)CSPs.

\begin{definition}
  Let $(\A,\B)$ be a PCSP template. A function $f:A^m\to B$ is a
  \emph{polymorphism} of arity $m$ of $(\A,\B)$ if for each pair of
  corresponding relations $R_i$ and $S_i$ from $\A$ and $\B$, respectively, the
  following holds: For any $(\ar(R_i)\times m)$ matrix $M$ whose columns
  are tuples in $R_i$, the application of $f$ to rows of $M$ gives a tuple in
  $S_i$. In other words, an arity $m$ polymorphism is a homomorphism from the $m$-th Cartesian power of $\A$ to $\B$. We denote by $\Pol(\A,\B)$ the set of all polymorphisms of $(\A,\B)$.
\end{definition}

In a $\PCSP$ template $(\A,\B)$ we view tuples from $\A$ and $\B$ as columns.
When writing tuples in text we may write them as rows to simplify notation but
they should still be understood as columns. For a $k$-ary relation $R$ on the set
$A$, we denote by $R^c=A^k \setminus R$ the complement of $R$. For a
relational structure $\A$, we denote by $\A^c$ the structure with relations
$R^c$ for each relation $R$ in $\A$. Most of our relational structures will be
on the Boolean domain $\{0,1\}$ and contain a single relation of arity $k$. The
(Hamming) weight of a tuple $\bx \in \{0,1\}^k$, denoted throughout by $d$, is
the number of 1's in $\bx$. For $1 \leq t < k$, the Boolean relational structure
$\inn{t}{k}$ consists (of one relation consisting) of all $k$-tuples with weight
$t$. The Boolean relational structure $\NAE$ contains all $k$-tuples except
$0^k$ and $1^k$. 

We need a definition and some notation to state existing results on Boolean (P)CSPs.

\begin{definition}\label{polysdef}
  A function $f:\{0,1\}^m \to \{0,1\}$ is 
  \begin{itemize}
  \item an $\OR_{m}$ ($\AND_{m}$) if it returns the logical OR (respectively logical AND) of its arguments;
  \item an alternating threshold $\AT_{m}$ if $m$ is odd and 
  \[f(x_1,\ldots,x_m)=1 \text{ if and only if } x_1-x_2+x_3-\cdots+x_{m}>0;\]
  \item a parity function $\XOR_{m}$ if $f(x_1,\ldots,x_m)=x_1+\cdots+x_m \mod 2$;
  \item a $q$-threshold $\THRo{q,m}$ (for $q$ a rational between 0 and 1 and $mq$
    not an integer) if $f(x_1,\ldots, x_m)=0$ if $\sum_{i=1}^m x_i < mq$ and $1$ otherwise;
  \item a majority $\MAJ_m$ if $f$ is a $\frac{1}{2}$-threshold and $m$ is odd.
  \end{itemize}

  We denote by $\OR$ and $\AND$ the set of all $\OR_m$ and $\AND_m$ functions, respectively, for all $m\geq 2$. We denote by $\AT$ and $\XOR$ the set of all $\AT_m$
  and $\XOR_m$ functions, respectively, for odd $m\geq 1$. Finally, $\THRo{q}$
  denotes the set of all $\THRo{q,m}$ functions for $qm \not\in \mathbb{Z}$.

  Define $\overline{f}$, the \emph{negation of $f$}, as the function $x \mapsto
  1-f(x)$, and for a family of functions $F$, define the \emph{negation of $F$} by $\overline{F}=\{\overline{f} | f \in F\}$.

\end{definition}

Schaefer's dichotomy theorem~\cite{s78} classified all Boolean CSP templates and can be stated in various forms (see e.g.~\cite{BKW17} for further discussion). Here we give a modern formulation in terms of polymorphisms.

\begin{theorem}
  \label{thm:schaefer}
  Let $\B$ be a Boolean CSP template. If $\Pol(\B)$ contains a constant,
  $\AND_2$, $\OR_2$, $\MAJ_3$, or $\XOR_3$, then $\CSP(\B)$ is tractable.
  Otherwise, $\CSP(\B)$ is NP-hard.
\end{theorem}

Ficak et al. classified all symmetric Boolean PCSP templates~\cite{fkos19}. 

\begin{theorem}[\cite{fkos19}] \label{thm:symmtrac}
  Let $(\A,\B)$ be a symmetric Boolean PCSP template. If $\Pol(\A,\B)$ contains
  a constant or at least one of $\OR$, $\AND$, $\XOR$, $\AT$, $\THRo{q}$ (for
  some $q$) or their negations, then $\PCSP(\A,\B)$ is tractable. Otherwise,
  $\PCSP(\A,\B)$ is NP-hard. 
\end{theorem}

The only possibly unresolved promise templates are those with NP-hard CSP templates.

\begin{proposition}\label{prop:cspeasy}
Let $(\A,\B)$ be a promise template such that at least one of $\CSP(\A)$, $\CSP(\B)$ is tractable. Then $\PCSP(\A,\B)$ is tractable. 
\end{proposition}

Proposition~\ref{prop:cspeasy} is a direct consequence of the important concept
of homomorphic relaxation, which we now define.
Let $(\A,\B)$ and $(\A',\B')$ be two PCSP templates over the same signature. We
call $(\A',\B')$ a \emph{homomorphic relaxation} of $(\A,\B)$ if $\A'\to\A$ and
$\B\to\B'$. It is easy to show~\cite{BBKO21} that
$\PCSP(\A',\B')\red\PCSP(\A,\B)$.\footnote{In fact, more is known: The trivial
(identity) reduction from $\PCSP(\A',\B')$ to $\PCSP(\A,\B)$ is correct if and
only if $(\A',\B')$ is a homomorphic relaxation of $(\A,\B)$.}

\begin{proof}[Proof of Proposition~\ref{prop:cspeasy}]
  We have $\PCSP(\A,\B)\red\PCSP(\A,\A)=\CSP(\A)$, since $(\A,\B)$ is a
  homomorphic relaxation of $(\A,\A)$ as $\A\to\B$ by assumption. Similarly,
  $\PCSP(\A,\B)\red\PCSP(\B,\B)=\CSP(\B)$, since $(\A,\B)$ is a homomorphic
  relaxation of $(\B,\B)$ as $\A\to\B$ by assumption.
\end{proof}

Theorem~\ref{thm:schaefer} established NP-hardness of two natural CSPs:
$\CSP(\inn{1}{3})$ and $\CSP(\NAE)$. Interestingly,
$\PCSP(\inn{1}{3},\NAE)$ is solvable in polynomial-time, as first shown by
Brakensiek and Guruswami~\cite{BG21}. 
(This shows that the converse of Proposition~\ref{prop:cspeasy} is false.)
A natural generalisation of $\inn{1}{3}$ is $\inn{t}{k}$.
Theorem~\ref{thm:schaefer} implies that $\CSP(\inn{t}{k})$ is NP-hard, which also follows from  Proposition~\ref{prop:tink}.
Theorem~\ref{thm:symmtrac} implies that the tractability of $\PCSP(\inn{1}{3},\NAE)$ also holds for 
$\PCSP(\inn{t}{k},\NAE)$. 

\begin{proposition}\label{prop:tinkhard}
  For $k \geq 3$ and $1 \leq t < k, \CSP(\inn{t}{k})$ is NP-hard.
\end{proposition}

\begin{proposition}\label{prop:naeeasy}
  For $k \geq 2$ and $1\leq t < k, \PCSP(\inn{t}{k},\NAE)$ is tractable.
\end{proposition}

\subsection{Algorithms}

We now present three relaxations for PCSPs: $\BLP$, $\AIP$, and $\BLP+\AIP$.
The first one, $\BLP$, is needed for the description of the third one.
The second one, $\AIP$, solves all tractable cases in our
classification results. Finally, the third one, $\BLP+\AIP$, is one of the
strongest known algorithms for PCSPs and the strongest one with a characterisation of
its power in terms of polymorphism identities. Our ``algorithmic dichotomy''
result (Theorem~\ref{thm:add}) shows $\AIP$ solvability vs. $\BLP+\AIP$-hardness.

In the rest of this section, let $(\A,\B)$ be a PCSP template, where $\A=(A;R_1,\ldots,R_p)$ and
$\B=(B;S_1,\ldots,S_p)$. Let $\X=(X;T_1,\ldots,T_p)$ be an instance of $\PCSP(\A,\B)$.
We assume without loss of generality that all three structures contain a
unary relation equal to $X$ in $\X$, equal to $A$ in $\A$, and equal to $B$ in
$\B$; the relation is called $R_u$ in $\A$. If this is not
the case, the template and the instance can be extended without changing the set
of solutions.

The \emph{basic linear programming relaxation} ($\BLP$) of $\X$, denoted by $\BLP(\X,\A)$,
is defined as follows. The variables are
$\lambda_{\bx,i}(\ba)$ for every $i\in [p]$, $\bx\in T_i$,
and $\ba\in R_i$, and the constraints are given in Figure~\ref{fig:blp}. (Note
that $\BLP(\X,\A)$ does not depend on $\B$.)
\begin{figure}[hbt]
\begin{align}
  0\ \leq\ \lambda_{\bx,i}(\ba)\ &\leq\ 1 & \forall i\in [p], \forall \bx\in T_i, \forall \ba\in R_i\\
  \sum_{\ba\in R_i} \lambda_{\bx,i}(\ba)\ &=\ 1 & \hspace*{1.5cm} \forall i\in [p], \forall \bx\in T_i\\
  \sum_{\ba\in R_i,a_j=a} \lambda_{\bx,i}(\ba)\ &=\ \lambda_{x_j,R_u}(a) & \hspace*{1.5cm} \forall i\in [p], \forall \bx\in T_i, \forall a\in A,\forall j\in [\ar(R_i)]
\end{align}
\caption{Definition of $\BLP(\A,\X)$.}
\label{fig:blp}
\end{figure}

The \emph{basic affine integer programming relaxation} ($\AIP$) of $\X$, denoted by
$\AIP(\X,\A)$, is defined as follows. The variables are
$\tau_{\bx,i}(\ba)$ for every
$i\in [p]$, $\bx\in T_i$, and $\ba\in R_i$, and the constraints are given in
Figure~\ref{fig:aip}.
\begin{figure}[hbt]
\begin{align}
  \tau_{\bx,i}(\ba)\ &\in\ \Z & \forall i\in [p], \forall \bx\in T_i, \forall \ba\in R_i\\
  \sum_{\ba\in R_i} \tau_{\bx,i}(\ba)\ &=\ 1 & \hspace*{1.5cm} \forall i\in [p], \forall \bx\in T_i\\
  \sum_{\ba\in R_i,a_j=a} \tau_{\bx,i}(\ba)\ &=\ \tau_{x_j,R_u}(a) & \hspace*{1.5cm} \forall i\in [p], \forall \bx\in T_i, \forall a\in A, \forall j\in [\ar(R_i)]
\end{align}
\caption{Definition of $\AIP(\A,\X)$.}
\label{fig:aip}
\end{figure}

\noindent We say that $\AIP(\X,\A)$ accepts if the affine program in
Figure~\ref{fig:aip} is feasible, and rejects otherwise. By construction, if
$\X\to\A$ then $\AIP(\X,\A)$ accepts. We say that $\AIP$ \emph{solves}
$\PCSP(\A,\B)$ if for every instance $\X$ accepted by $\AIP(\X,\A)$ we have
$\X\to\B$. 

A $(2m+1)$-ary function $f:A^{2m+1}\to B$ is called \emph{alternating} if
$f(a_1,\ldots,a_{2m+1})=f(a_{\pi(1)},\ldots,\allowbreak a_{\pi(2m+1)})$ for every
$a_1,\ldots,a_{2m+1}\in A$ and every permutation $\pi:[2m+1]\to[2m+1]$ that preserves
parity, and $f(a_1,\ldots,a_{2m-1},a,a)=f(a_1,\ldots,a_{2m-1},a',a')$ for every
$a_1,\ldots,a_{2m-1},a,a'\in A$. Intuitively, an alternating function is
invariant under permutations of its odd and even coordinates and has the
property that adjacent coordinates cancel each other out.
The power of $\AIP$ for PCSPs is characterised by the following
result.\footnote{We note that~\cite{BBKO21} proves several other equivalent
statements in Theorem~\ref{thm:aip}.}

\begin{theorem}[\cite{BBKO21}]\label{thm:aip}
  Let $(\A,\B)$ be a PCSP template. Then (the decision version of)
  $\PCSP(\A,\B)$ is tractable via $\AIP$ if and only if $\Pol(\A,\B)$ contains
  alternating functions of all odd arities.
\end{theorem}

The \emph{combined basic LP and affine IP algorithm} ($\BLP+\AIP$)~\cite{BGWZ20}
is presented in Algorithm~\ref{alg:blp-aip}.
\begin{algorithm}[tbh] 
	\SetAlgoLined
  \KwIn{\quad\ \ \ an instance $\X$ of $\PCSP(\A,\B)$}
  \KwOut{\quad \textsc{yes} if $\X\to\A$ and \textsc{no} if $\X\not\to\B$}
  \medskip
  find a relative interior point $(\lambda_{\bx,i}(\ba))_{i\in [p], \bx\in T_i,\ba\in R_i}$ of $\BLP(\X,\A)$\;
  \If{no relative interior point exists}{return \textsc{no}\;}
  refine $\AIP(\X,\A)$ by setting $\tau_{\bx,i}(\ba)=0$ if $\lambda_{\bx,i}(\ba)=0$\;
  \If{the refined $\AIP(\X,\A)$ accepts}{return \textsc{yes}\;}{return \textsc{no}\;}
\caption{The $\BLP+\AIP$ algorithm} \label{alg:blp-aip}
\end{algorithm}
If $\X\to\A$ then $\BLP+\AIP$ accepts $\X$~\cite{BGWZ20}. We say that $\BLP+\AIP$
\emph{solves} $\PCSP(\A,\B)$ if for every instance $\X$ accepted by $\BLP+\AIP$
we have $\X\to\B$. 

A $(2m+1)$-ary function $f:A^{2m+1}\to B$ is called \emph{$2$-block-symmetric}
if $f(a_1,\ldots,a_{2m+1})=f(a_{\pi(1)},\ldots,\allowbreak a_{\pi(2m+1)})$ for
every $a_1,\ldots,a_{2m+1}\in A$ and every permutation $\pi:[2m+1]\to[2m+1]$
that preserves parity. In other words, $f$ is $2$-block-symmetric if its $2m+1$
coordinates can be partitioned into two blocks of size $m+1$ and $m$ in such a
way that the value of $f$ is invariant under any permutation of coordinates
within each block. Without loss of generality, we will assume that the two
blocks are the odd and even coordinates of $f$.

The power of $\BLP+\AIP$ for PCSPs is characterised by the following
result.\footnote{We note that~\cite{BGWZ20} proves several other equivalent
statements in Theorem~\ref{thm:blp-aip}.}

\begin{theorem}[\cite{BGWZ20}]\label{thm:blp-aip}
  Let $(\A,\B)$ be a PCSP template. Then (the decision version of)
  $\PCSP(\A,\B)$ is tractable via $\BLP+\AIP$ if and only if $\Pol(\A,\B)$ contains
  2-block-symmetric functions of all odd arities.
\end{theorem}

\section{Results}
\label{sec:results}

Our results are concerned with templates that arise from $(\inn{t}{k},\NAE)$ either by adding tuples to
$\inn{t}{k}$ or removing tuples from $\NAE$. For a set of tuples
$S\subseteq\{0,1\}^k$, we write $\inn{t}{k}\cup\Ss$ for the relational structure
whose (only) relation contains all $k$-tuples of weight $t$ and the tuples
from $S$, and similarly for $\NAE\setminus\Ss$.

Our first result is an algorithmic dichotomy for templates constructed by adding tuples to
$\inn{t}{k}$.

\begin{theorem}[\textbf{Main \#1}]\label{thm:add}
  Let $k \geq 3$ and $\emptyset\neq S \subseteq (\inn{t}{k})^c \cap \NAE$.
  If $t$ is odd, $k$ is even, and $S$ contains tuples of only odd 
  weight, then $\PCSP(\inn{t}{k}\cup\Ss,\NAE)$ is tractable via $\AIP$. Otherwise,
  $\PCSP(\inn{t}{k}\cup\Ss,\NAE)$ is not solved by $\BLP+\AIP$.
\end{theorem}

Ruling out the applicability of $\BLP+\AIP$ as stated in Theorem~\ref{thm:add}
is done, via Theorem~\ref{thm:blp-aip}, in Section~\ref{sec:t}.

\begin{remark}\label{rm:intract}
Barto et al.~\cite{BBKO21} showed that $\PCSP(\inn{1}{3},\NAE)$ is not ``finitely tractable'',
meaning that there is no finite $\C$ such that $\inn{1}{3}\to\C\to\NAE$ and
$\CSP(\C)$ is tractable. In other words, the tractability of
$\PCSP(\inn{1}{3},\NAE)$ cannot be achieved via a ``gadget reduction'' to
tractable finite-domain CSPs. This result was then extended by Asimi and
Barto~\cite{Asimi21:mfcs} to $\PCSP(\inn{t}{k},\NAE)$ for $k\geq 3$, $t<k$
when $t$ is even or $k$ is odd. Since the $\BLP+\AIP$-hard cases in
Theorem~\ref{thm:add} are homomorphically sandwiched by templates proved 
finitely intractable in~\cite{Asimi21:mfcs}, they are also finitely
intractable.

A recent result of Atserias and Dalmau that gives a necessary condition for
PCSPs to be solvable by a ``local consistency checking''
algorithm~\cite{Atserias22:soda} implies that all templates from
Theorem~\ref{thm:add} (and in particular those not solved by $\BLP+\AIP$) are \emph{not} solved by a ``local
consistency checking'' algorithm. By \cite[Corollary~4.2]{Atserias22:soda}, such an algorithm does not solve
$\PCSP(\inn{t}{k},\NAE)$ for any $k\geq 3$ and $t<k$, and since $(\inn{t}{k},\NAE)$ is a homomorphic relaxation of the templates from
Theorem~\ref{thm:add}, our claim follows from \cite[Lemma 7.5]{BBKO21}.

\end{remark}

Our second result is a complexity dichotomy for templates constructed by
removing tuples from $\NAE$. The key result here is the following.

\begin{theorem}\label{thm:t-in-k}
  Let $k \geq 3$ and let $\T \subseteq \{0,1\}^k$ be a relation such that $\inn{t}{k} \to \T$ and $\CSP(\T)$ is NP-hard. Then
  $\PCSP(\inn{t}{k},\T)$ is tractable if and only if $\T=\NAE$, unless P=NP.
\end{theorem}

In other words, $\PCSP(\inn{t}{k},\T)$ is tractable if $\CSP(\T)$ is tractable
or $\T=\NAE$, and is NP-hard otherwise.
Theorem~\ref{thm:t-in-k} then easily implies the following.

\begin{theorem}[\textbf{Main \#2}]\label{thm:rem}
  Let $k \geq 3$ and $\emptyset\neq S \subseteq (\inn{t}{k})^c \cap \NAE$.
  If $t$ is odd, $k$ is even, and $S$ contains tuples of only even 
  weight, then $\PCSP(\inn{t}{k},\NAE\setminus\Ss)$ is tractable. Otherwise,
  $\PCSP(\inn{t}{k},\NAE\setminus\Ss)$ is NP-hard.
\end{theorem}

Theorem~\ref{thm:t-in-k} is proved in Section~\ref{sec:rem} and relies on
Theorem~\ref{thm:symmtrac}, as well as a symmetrisation trick
(Proposition~\ref{prop:firstsym}, observed independently
in~\cite{Barto21:stacs}) and the following observation.

\begin{proposition}\label{prop:symimage}
  Let $R$ be a symmetric relation on a set $A$. For any
  function $f:A \to B$, the component-wise image of $R$ under $f$, denoted $f(R)$, is a symmetric relation on $B$.
\end{proposition}

\begin{proof}
  Suppose that $\by \in f(R)$, so $\by=f(\bx)$ for some $\bx \in R$. We must
  show that $\pi(\by) \in f(R)$ for an arbitrary permutation $\pi$. But since
  $R$ is symmetric, we have $\pi(\bx) \in R$, and so $f(\pi(\bx))=\pi(\by)$
  since $f$ is applied component-wise.
\end{proof}

\begin{proposition}\label{prop:firstsym}
  Let $(\A,\B)$ be a PCSP template with $\A$ symmetric. For each relation $R \in \B$, let $R'$ be the largest symmetric relation contained in $R$. Let $\B'$ be the relational structure with the same domain as $\B$ but with relations $R'$ instead of $R$. Then $\PCSP(\A,\B)$ is polynomial-time equivalent to $\PCSP(\A,\B').$
\end{proposition}

\begin{proof}
  We first check that $(\A,\B')$ is a valid PCSP template, i.e., that there is a homomorphism $\A \to \B'$. Let $\phi$ be a homomorphism from $\A$ to $\B$. By Proposition~\ref{prop:symimage}, $\phi(\A)$ is symmetric, and since $\B'$ is the largest symmetric relational structure contained in $\B$, we have $\phi(\A) \subseteq \B'$. Therefore $(\A,\B')$ is a valid PCSP template.
  For $f \in \Pol(\A,\B)$, $f(\A)$ is symmetric and is contained in $\B$, so $f(\A) \subseteq \B'$ and $\Pol(\A,\B) \subseteq \Pol(\A,\B')$. The reverse inclusion follows from $\B' \subseteq \B$ and gives $\Pol(\A,\B) = \Pol(\A,\B')$, which implies by~\cite[Theorem 3.1]{BBKO21} that $\PCSP(\A,\B) \equiv_p \PCSP(\A,\B')$.
\end{proof}

The tractability parts in Theorem~\ref{thm:add} and Theorem~\ref{thm:rem} follow easily from existing work, as we
now show.
Using (the sufficiency of) Theorem~\ref{thm:aip}, it is easy to establish 
Proposition~\ref{prop:naeeasy}, i.e., tractability of $\PCSP(\inn{t}{k},\NAE)$.
Indeed, by~\cite[Claim~4.6]{BG21}, the $\AT$ family maps collections of
$\inn{t}{k}$ tuples into $\NAE$, and so $\PCSP(\inn{t}{k},\NAE)$ is tractable.

\begin{proposition}
  \label{proposition:oddinktract}
  For $k$ even, (the search version of) $\CSP(\odd{k})$ is tractable.
\end{proposition}

\begin{proof}
  We claim that $\XOR_3 \in \Pol(\odd{k})$ so that tractability will follow from
  Theorem~\ref{thm:schaefer}. To see that $\XOR_3 \in \Pol(\T)$, suppose that
  $\XOR_3$ returns a tuple of even  weight $d$. Then in the $k \times 3$
  matrix of inputs with three odd  weight tuples as columns, there are
  $d$ rows with an odd number of 1's and $k-d$ rows with an even number of 1's.
  Together these give an even total number of 1's in the matrix. But since the
  three input columns have odd weight, the total number of 1's in the matrix is odd. Contradiction.
\end{proof}

\begin{proof}[Proof of the tractability part of Theorems~\ref{thm:add} and~\ref{thm:rem}]
Under the tractability criterion of
Theorem~\ref{thm:add}, $\inn{t}{k}\cup\Ss\subseteq\odd{k}\subseteq\NAE$ and thus
$(\inn{t}{k}\cup\Ss,\NAE)$ is a homomorphic relaxation of $(\odd{k},\odd{k})$.
As discussed in Section~\ref{sec:prelims} (and proved in~\cite{BBKO21}), this
implies that
$\PCSP(\inn{t}{k}\cup\Ss,\NAE)\leq_p\PCSP(\odd{k},\odd{k})=\CSP(\odd{k})$, where
$\CSP(\odd{k})$ is tractable by Proposition~\ref{proposition:oddinktract} (for even $k$). Similarly, under the
tractability criterion of Theorem~\ref{thm:rem}, we have
$\inn{t}{k}\subseteq\odd{k}\subseteq\NAE\setminus\Ss$ and thus
$\PCSP(\inn{t}{k},\NAE\setminus\Ss)\leq_p\CSP(\odd{k})$. 
By composing $\XOR_3$ functions from the proof of
  Proposition~\ref{proposition:oddinktract}, or by observing that $\odd{k}$ is
  an affine subspace, we have $\XOR \subseteq \Pol(\odd{k})$, which implies via
  the inclusion that our PCSP templates  have the $\XOR$ family of polymorphisms
  and thus are solvable by $\AIP$.
\end{proof}

\section{Adding tuples}\label{sec:t}

The following result implies, by Theorem~\ref{thm:blp-aip}, the non-tractability part of
Theorem~\ref{thm:add}.

\begin{theorem}\label{thm:gen-t}
Let $k\geq3$, $1\leq t<k$, and $\bx$ be a $k$-tuple of weight $1\leq d<k$ with $d\neq t$. 
Then, $(\inn{t}{k}\cup\{\bx\},\NAE)$ does not have 2-block-symmetric polymorphisms of
all odd arities, unless $t$ is odd, $k$ is even, and $d$ is odd.
\end{theorem}

The implication is as follows: In the non-tractability case of Theorem~\ref{thm:add}, $\Ss$ contains a tuple $\bx$ of weight $d$ such that if $d$ is odd, then $t$ is even, $k$ is odd, or both. Therefore $\Pol(\inn{t}{k}\cup \Ss,\NAE)\subseteq \Pol(\inn{t}{k}\cup \{\bx\},\NAE)$, so it suffices to rule out 2-block-symmetric polymorphisms for templates of the form $(\inn{t}{k}\cup \{\bx\},\NAE)$.

We start with two simple observations which reduce the number of cases to deal with. Since permuting the rows of a
matrix of inputs to a polymorphism permutes the values of the output
tuple and does not affect membership in the symmetric $\NAE$
relation, we have the following.

\begin{observation}\label{obs:tuple}
  Let $\bx$ and $\by$ be two $k$-tuples of weight $d$. 
  Then, $\Pol(\inn{t}{k}\cup \{\bx\},\NAE) = \Pol(\inn{t}{k}\cup \{\by\},\NAE)$. 
\end{observation}

By Observation~\ref{obs:tuple}, it suffices to prove Theorem~\ref{thm:gen-t} for $\bx$ of the form $\bx=1^d0^{k-d}$.

\begin{observation}\label{obs:invert}
 There is a bijection between $\Pol(\inn{t}{k}\cup \{\bx\},\NAE)$ and $\Pol(\inn{(k-t)}{k}\cup \{\overline{\bx}\},\NAE)$ given by $f(x_1,\ldots,x_m) \mapsto f(1-x_1,\ldots,1-x_m)$, where $\overline{\bx}$ is the negation of $\bx$. 
\end{observation}

Observation~\ref{obs:invert} implies that 2-block-symmetry of polymorphisms is preserved when swapping 0's and 1's.

There are eight combinations of the parities of $k$, $t$, and $d$. The case $(k,t,d) \equiv (0,1,1)$ is out of the scope of
Theorem~\ref{thm:gen-t} and is covered by the tractable case of
Theorem~\ref{thm:add}. The case $(k,t,d) \equiv (1,0,0)$ is covered for $d>t$ in Proposition~\ref{prop:d>t-even} and $d<t$ in Proposition~\ref{prop:d<t-even}, and all other cases are covered for $d>t$ in Proposition~\ref{prop:d>t} and $d<t$ in Proposition~\ref{prop:d<t}. By applying Observation~\ref{obs:invert}, we may assume that $d+t \leq k$, which allows a single construction to work in each of these propositions. We start with a brief account of the idea behind the proofs.

 \medskip

  Let $C_k^t$ be the $k \times k$ matrix containing the $k$ cyclic shifts of
  the column $1^t0^{k-t}$. The matrix $C_k^t$ can be used to fill one of the
  coordinate blocks of a 2-block-symmetric function $f$ of arity $2k\pm 1$. For
  example, suppose that $C_k^t$ is used to fill the ``first'' coordinate block.
  It does not matter whether the first block contains the odd or even
  coordinates. Then $f$ depends only on the weights in each row of the
  ``second'' block, since the first block has the same weight in every row. This
  allows $f$ to be analysed as a symmetric (1-block-symmetric) function. 
  
  For each $k$, $t$, and $d$, and for any $f$ such that one of its blocks can be filled by $C_k^t$, we exhibit a set of tableaux for the other
  block that prevents $f$ from being a polymorphism. Suppose we have filled one
  block with $C_k^t$, so that $f$ can now be represented
  as a unary function of the weight on its other block. For any weights
  $w_1,w_2,w_3$, we have at least one of $f(w_1)=f(w_2)$, $f(w_1)=f(w_3)$, and
  $f(w_2)=f(w_3)$. For each pair of weights, we construct a tableau
  where each row of the second block is one of the two weights. Thus we are guaranteed that $f$ will return an
  all-equal tuple and hence not be a polymorphism. 
  
\begin{proposition}\label{prop:d>t} 
  Let $k \geq 3$, $1 \leq t < k$, and $t < d <
  k$ be such that $d+t \leq k$, and $t\equiv k$ or $d\not\equiv t \pmod{2}$. Let $\bx$ be a tuple of weight $d$. Then
  $\Pol(\inn{t}{k}\cup \{\bx\},\NAE)$ does not have 2-block-symmetric functions of
  arity $2k-1$.
\end{proposition}

\begin{proof}
  Let $f$ be a 2-block-symmetric function of arity $2k-1$. We will show that $f$
  is not a polymorphism of $(\inn{t}{k}\cup\{\bx\},\NAE)$. Let the odd block
  contain the tableau $C_k^t$ and denote by
  $C_k^{t-}$ the tableau obtained from $C_k^{t}$ by removing its last column. We describe how to construct the even
  block with $k-1$ columns in the cases below. We use $t-1$, $t$, and $t+1$ as our three weights.

  \noindent \textbf{Case 1}: weights $t-1$ and $t$. 

  \noindent The even tableau is $C_k^{t-}$. Thus each row in the even block has weight
  either $t-1$ or $t$. If $f(t-1)=f(t)$ then $f$ returns an all-equal tuple
  $0^k$ or $1^k$, so $f\not\in\Pol(\inn{t}{k}\cup\{\bx\},\NAE)$.

  \noindent \textbf{Case 2}: weights $t-1$ and $t+1$.

  \noindent We take $C_k^{t-}$  and replace some of the $\inn{t}{k}$ tuples with
  $\bx$ as necessary. 

  \noindent \textbf{Case 2a}: $t \equiv k \pmod{2}$.

  \noindent The tableau $C_k^{t-}$ has an even number $k-t$ of rows with weight
  $t$. These can be paired up and 1's exchanged so that each row has weight
  either $t-1$ or $t+1$. In particular, in the columns $t+1,t+3,\ldots,k-1$, we
  swap the values in the pairs of rows $(t,t+1),(t+2,t+3),\ldots,(k-2,k-1)$,
  respectively. An illustration is given in Figure~(\ref{fig:case2a}).

 \begin{figure}[hbt]
  \centering
  \begin{minipage}[b]{.4\linewidth}
  \centering
    \[ 
      \left(\begin{array}{cccccccc}
        1&0&0&0&0&0&0&1 \\ 
        1&1&0&0&0&0&0&0 \\
        1&1&1&\mathbf{1}&0&0&0&0 \\
        0&1&1&\mathbf{0}&0&0&0&0 \\ 
        0&0&1&1&1&\mathbf{1}&0&0 \\
        0&0&0&1&1&\mathbf{0}&0&0 \\
        0&0&0&0&1&1&1&\mathbf{1} \\
        0&0&0&0&0&1&1&\mathbf{0} \\
        0&0&0&0&0&0&1&1 \\
    \end{array}\right)
   \] 
  \subcaption{Case 2a with $t=3$.\label{fig:case2a}}
  \end{minipage}
  \qquad \qquad
  \begin{minipage}[b]{.4\linewidth}
     \centering
    \[ 
  \left(\begin{array}{cccccccc}
        1&0&0&0&0&0&0&0 \\ 
        1&\mathbf{0}&0&0&0&0&0&0 \\
        1&1&1&0&0&0&0&0 \\
        1&0&1&1&0&0&0&0 \\ 
        1&0&0&1&1&0&0&0 \\
        0&\mathbf{1}&0&0&1&1&0&0 \\
        0&0&0&0&0&1&1&\mathbf{1} \\
        0&0&0&0&0&0&1&\mathbf{0} \\
        0&0&0&0&0&0&0&1 \\
    \end{array}\right)
    \] 
    \subcaption{Case 2b with $t=2$ and $d=5$.\label{fig:case2b}}
  \end{minipage}
  \caption{Example with $k=9$ for $f(t-1)=f(t+1)$. Swapped values in bold.}
  \end{figure}

  \noindent \textbf{Case 2b}: $(k,t,d) \equiv (0,1,0) $ or $ (1,0,1) \pmod{2}$.
  
  \noindent The tableau $C_k^{t-}$ has an odd number $k-t$ of rows with weight
  $t$. We replace the first column with $\bx$. This adds $d-t$ 1's to the first
  column, and the $d-t$ rows with the added 1's now have weight $t+1$. There
  remains an even number $k-d$ of rows of weight $t$, which can be paired up
  to exchange 1's and achieve weight $t-1$ or $t+1$ in each row. In particular,
  we swap the values at positions $(t,2)$ and $(d+1,2)$, and then in the columns
  $d+3,d+5,\ldots,k-1$, we swap the values in the pairs of rows
  $(d+2,d+3),(d+4,d+5),\ldots,(k-2,k-1)$, respectively. An illustration is given
  in Figure~(\ref{fig:case2b}).

  Cases 2a and 2b cover all possible parities of $k$, $t$, and $d$ under the
  proposition's assumptions. In both cases, each row in the even block has weight either $t-1$ or $t+1$. If $f(t-1)=f(t+1)$ then $f$ returns an all-equal tuple, so $f\not\in\Pol(\inn{t}{k}\cup\{\bx\},\NAE)$.

 \noindent  \textbf{Case 3}: weights $t$ and $t+1$.

  We first give a general description of the tableau, and then derive values for its parameters. We place the tuple $\bx$ in the first $r$ columns and fill the remaining $k-1-r$ columns with $\inn{t}{k}$ tuples that have specific behaviours in the upper $d$ rows and  lower $k-d$ rows. Our goal is to distribute the weight of the $\inn{t}{k}$ tuples between these two blocks of rows so that every row in the full tableau has weight $t$ or $t+1$. 
  
  Denoting by $a$ the average column weight within the upper group of rows, we place either $b$ or $b+1$ 1's from each $\inn{t}{k}$ tuple in the upper block, where $b$ is an integer close to $a$. More precisely, to fill the upper $d$ rows, we use $0 \leq s \leq k-1-r$ columns of weight $b$ and $k-1-r-s$ columns of weight $b+1$, and in the lower block of rows, we use $s$ columns of weight $t-b$ and $k-1-r-s$ columns of weight $t-(b+1)$, respectively, so that the full columns are $\inn{t}{k}$ tuples.

  We now describe the position of the 1's in the upper group of rows; the construction for the lower group is analogous. We fill columns with 1's from top to bottom, starting at the left-most column, and moving to the right after placing a column's quota of 1's (either $b$ or $b+1$). The order of the weight $b$ and $b+1$ tuples does not matter. When moving right to the next column, we continue placing 1's in the row immediately below the lowest row containing a 1 in the previous column. Once we reach the bottom of the group of rows, we wrap around to the top and continue in this way.

  A $k \times (k-1)$ tableau containing only $\inn{t}{k}$ tuples needs at least $t$ more 1's to achieve weight $\geq t$ in each row. Each time we replace a $\inn{t}{k}$ tuple with $\bx$, the
  tableau gains $d-t$ 1's, and therefore at least $r=\left\lceil\frac{t}{d-t}\right\rceil$ occurrences of $\bx$ are necessary. This turns out to be sufficient. The remainder of the proof is devoted to showing that there exist $a$, $b$, and $s$ which allow our construction to work. 
  
  A crucial observation connecting the column and row weights in our tableau is that within each block of rows, the weight between rows varies by at most one. It therefore suffices for the \emph{average} row weight in each block to be between $t$ and $t+1$. 

  To achieve average row weight at least $t$, $a$ must be such that the total weights in the upper and lower blocks are at least $d(t-r)$ and $t(k-d)$, respectively. This is guaranteed when both $a(k-1-r) \geq d(t-r)$ and $(t-a)(k-1-r) \geq t(k-d)$, or equivalently,
  
  \begin{equation}
    \label{eqn:first}
  \frac{d(t-r)}{k-1-r}\ \leq\ a\ \leq\ \frac{t(d-r-1)}{k-1-r}.
  \end{equation}

  To achieve average row weight at most $t+1$, $a$ must be such that the total weights in the upper and lower blocks are at most $d(t+1-r)$ and $(k-d)(t+1)$, respectively. This is guaranteed when both $a(k-1-r) \leq d(t+1-r)$ and $(t-a)(k-1-r) \leq (k-d)(t+1)$, or equivalently,  

  \begin{equation}
      \label{eqn:second}
    \frac{t(d-r-1)-k+d}{k-1-r}\ \leq\ a\ \leq\ \frac{d(t-r+1)}{k-1-r}.
  \end{equation}

  Finally, to ensure that each block of rows is tall enough to accommodate the 1's specified in the construction, it suffices that $0 \leq a \leq d$ and $0 \leq t-a \leq k-d$, which together are equivalent to $\max(0,d+t-k) \leq a \leq \min(d,t)$. By our assumptions, this reduces to $0 \leq a \leq t$.

  We now show that these inequalities can all be simultaneously satisfied. In~\ref{eqn:first}, the upper bound is at least the lower bound if and only if $r \geq \frac{t}{d-t}$, which holds for our choice of $r$, and in~\ref{eqn:second}, the upper bound is at least the lower bound if and only if $r \leq  \frac{t}{d-t} +\frac{k}{d-t}$, which holds since
  $\frac{k}{d-t}\geq 1$. Exchanging the upper/lower bound pairs in~\ref{eqn:first} and~\ref{eqn:second} results in two pairs of inequalities on $a$ that are always satisfied.   Finally, for $0 \leq a \leq t$, it suffices that at least one of the lower bounds is nonnegative, and at least one of the upper bounds is at most $t$. The lower bound in~\ref{eqn:first} is nonnegative if and only if $t \geq r$, which always holds, and the upper bound is at most $t$ if and only if $d \leq k$, which also always holds. 

  Therefore there exists $0 \leq a \leq t$ satisfying~\ref{eqn:first} and~\ref{eqn:second}, and since these inequalities are not strict, we can take $a$ to be rational with denominator $k-1-r$. Let $b=\lfloor a \rfloor$. Recalling that $s$ is the number of columns of weight $b$ in the upper block, computing the total weight in the upper block gives $sb+(k-1-r-s)(b+1)=a(k-1-r)$, so that $s=(k-1-r)(b+1-a)$. This is an integer since $a$ is a fraction with denominator $k-1-r$. As a sanity check, note that if $a=b$ or $a=b+1$, then $s=k-1-r$ or $s=0$, respectively. 
  
  We have shown that there exist $a$, $b$, and $s$  which permit us to construct the tableau with weight $t$ or $t+1$ in each row. Thus if $f(t)=f(t+1)$, then $f$ returns the all-equal tuple $0^k$ or $1^k$, so $f\not\in\Pol(\inn{t}{k}\cup\{\bx\},\NAE)$. This ends
  the proof of Case~3. 
  
  Since we must have at least one of $f(t-1)=f(t)$, $f(t-1)=f(t+1)$, and $f(t)=f(t+1)$, the three cases complete the proof.

  An example with $t=7$, $k=15$, and $d=10$ is illustrated in
  Figure~\ref{fig:d>tcase3}. In this case, we have $r=\left\lceil\frac{t}{d-t}\right\rceil=3$ and we
  get the inequalities $\frac{40}{11} \leq a \leq \frac{42}{11}$ and
  $\frac{37}{11}\leq a \leq\frac{50}{11}$. We take $a=\frac{41}{11}$; the
  values $\frac{40}{11}$ and $\frac{42}{11}$ would also work. Then $b=3$ and
  $s=3$, so in the upper group we have 3 columns of weight $b=3$ and 8 columns
  of weight $b+1=4$. The columns containing $\bx$ are shown in addition to the
  construction on $k-1-r$ columns.

\begin{figure}[hbt]
  \centering
  \[
    \left(\begin{array}{ccc|ccccccccccc}
      1&1&1&1&~&~&1&~&1&~&~&1&~&1 \\ 
      1&1&1&1&~&~&1&~&~&1&~&1&~&~ \\
      1&1&1&1&~&~&1&~&~&1&~&1&~&~ \\
      1&1&1&~&1&~&~&1&~&1&~&~&1&~ \\
      1&1&1&~&1&~&~&1&~&1&~&~&1&~ \\
      1&1&1&~&1&~&~&1&~&~&1&~&1&~ \\
      1&1&1&~&~&1&~&1&~&~&1&~&1&~ \\
      1&1&1&~&~&1&~&~&1&~&1&~&~&1 \\
      1&1&1&~&~&1&~&~&1&~&1&~&~&1 \\
      1&1&1&~&~&~&1&~&1&~&~&1&~&1 \\ \hline
      ~&~&~&1&1&1&~&1&1&~&1&~&1&1 \\
      ~&~&~&1&1&1&~&1&~&1&1&~&1&~ \\
      ~&~&~&1&1&~&1&1&~&1&~&1&1&~ \\
      ~&~&~&1&~&1&1&~&1&1&~&1&~&1 \\
      ~&~&~&~&1&1&1&~&1&~&1&1&~&1 \\
  \end{array}\right)
  \]
  \caption{Example with $t=7$, $k=15$, and $d=10$ from Case 3, for $f(t)=f(t+1)$. \label{fig:d>tcase3}}
  \end{figure}
\end{proof}

\begin{proposition}\label{prop:d<t} 
  Let $k \geq 3$, $1 < t < k$, and $1 \leq d <
  t$ be such that $d+t \leq k$, and $t\equiv k$ or $d\not\equiv t \pmod{2}$. Let $\bx$ be a tuple of weight $d$. Then
  $\Pol(\inn{t}{k}\cup \{\bx\},\NAE)$ does not have 2-block-symmetric functions of
  arity $2k+1$.
\end{proposition}

\begin{proof}
  The proof is similar to the case $d>t$ established in
  Proposition~\ref{prop:d>t}, except that now we can reduce the number of 1's in the
  tableaux by replacing $\inn{t}{k}$ tuples with $\bx$. 
  We place the tableau $C_k^t$ in the even coordinates, so that there
  are $k+1$ columns to be filled in the odd coordinates. As before, we
  give tableaux for the three pairs of weights from $t-1$, $t$, and $t+1$.

  Let $C_k^{t+}$ be the $k \times (k+1)$ matrix $C_k^t$ with an extra column $1^t0^{k-t}$. 

  \noindent \textbf{Case 1}: weights $t$ and $t+1$. 

  \noindent The odd tableau is $C_k^{t+}$, so each row in the odd block has weight either $t$ or $t+1$.

  \noindent \textbf{Case 2}: weights $t-1$ and $t+1$.

  \noindent The tableaux are similar to Case~2 in the proof
  of Proposition~\ref{prop:d>t}. When $t \equiv k$, we
  modify $C_k^{t+}$ in the columns $t+2,t+4,\ldots,k$ by swapping the values in
  the pairs of rows $(t+1,t+2),(t+3,t+4),\ldots,(k-1,k)$, respectively. When
  $(k,t,d) \equiv (0,1,0) $ or $(1,0,1)$, we replace the first column of
  $C_k^{t+}$ with $\bx$, which leaves an even number $k-d$ of rows of weight
  $t$. Therefore in columns $d+2,d+4,\ldots,k$ we swap the values in the pairs
  of rows $(d+1,d+2),(d+3,d+4),\ldots,(k-1,k)$, respectively, to get weight $t-1$ or $t+1$ in each row. 

  \noindent \textbf{Case 3}: weights $t-1$ and $t$.

  \noindent This case is similar to Case 3 in the proof of Proposition~\ref{prop:d>t}. The tableau $C_k^{t+}$ has $t$ rows
  with weight $t+1$ and $k-t$ rows with weight $t$, so we must reduce the total
  weight by at least $t$. Replacing a $\inn{t}{k}$ tuple with $\bx$ reduces
   the weight by $t-d$, which suggests $r=\left\lceil\frac{t}{t-d}\right\rceil$ such replacements.

  Let $a$ be the average column weight in the upper $d$ rows. To achieve average row weight at most $t$, $a$ must be such that the total weights in the upper and lower blocks are at most $d(t-r)$ and $t(k-d)$, respectively. This is guaranteed when both $a(k+1-r) \leq d(t-r)$ and $(t-a)(k+1-r) \leq t(k-d)$, or equivalently,
  
  \begin{equation}
    \label{eqn:first-d<t}
   \frac{t(d-r+1)}{k+1-r} \leq a \leq \frac{d(t-r)}{k+1-r}.
  \end{equation}

  To achieve average row weight at least $t-1$, $a$ must be such that the total weights in the upper and lower blocks are at least $d(t-1-r)$ and $(k-d)(t-1)$, respectively. This is guaranteed when both $a(k+1-r) \geq d(t-1-r)$ and $(t-a)(k+1-r) \geq (k-d)(t-1)$, or equivalently,  

  \begin{equation}
  \label{eqn:second-d<t}
  \frac{d(t-1-r)}{k+1-r}\ \leq\ a\ \leq\ \frac{t(d-r+1)+k-d}{k+1-r}.
  \end{equation}

  Finally, to ensure that each block of rows is tall enough to accommodate the 1's specified by the construction, it suffices that $0 \leq a \leq d$ and $0 \leq t-a \leq k-d$, which together are equivalent to $\max(0,d+t-k) \leq a \leq \min(d,t)$. By our assumptions, this reduces to $0 \leq a \leq d$.

  We now show that these inequalities can all be simultaneously satisfied. In~\ref{eqn:first-d<t}, the upper bound is at least the lower bound if and only if $r \geq \frac{t}{t-d}$, which holds for our choice of $r$. In~\ref{eqn:second-d<t}, the upper bound is at least the lower bound if and only if $r \leq \frac{t}{t-d}+\frac{k}{t-d}$, which holds since $\frac{k}{t-d}\geq 1$. Exchanging the upper/lower bound pairs in~\ref{eqn:first-d<t} and~\ref{eqn:second-d<t} results in two pairs of inequalities on $a$ that are always satisfied. Finally, for $0 \leq a \leq d$, it suffices that at least one of the lower bounds is nonnegative, and at least one of the upper bounds is at most $d$. The lower bound in~\ref{eqn:first-d<t} is nonnegative if and only if $t \geq d+1$, and the upper bound is at most $d$ if and only if $t \leq k+1$, both of which always hold. The rest of the proof follows the same reasoning as in Proposition~\ref{prop:d>t}.
\end{proof}

\begin{proposition}\label{prop:d>t-even}
  Let $k \geq 3$, $1 \leq t < k$, and $t < d < k$ be such that $d+t \leq k$ and $(k,t,d)\equiv (1,0,0) \pmod{2}$. Let $\bx$ be a tuple of weight $d$. Then $\Pol(\inn{t}{k}\cup \{\bx\},\NAE)$ does not have 2-block-symmetric functions of all odd arities.
\end{proposition}

\begin{proof}
  The parities of $k$, $t$, and $d$ prevent us from using the weights
  $t-1$ and $t+1$ in Case~2, which necessitates a different choice of weights and a slightly more complicated construction for Case 3.
  We take $L\geq 1$ copies of $C_k^t$ in the even block of the tableau, and leave
  $Lk+1$ columns to be filled in the odd block, for a total arity of $2Lk+1$. The
  three weights we use are $Lt$, $Lt+1$, and $Lt+2$, with $L$ determined later. 

  \noindent \textbf{Case 1}: weights $Lt$ and $Lt+1$.

  \noindent We take $L-1$ copies of $C_k^t$ and one copy of $C_k^{t+}$, so that each row has weight $Lt$ or $Lt+1$.

  \noindent \textbf{Case 2}: weights $Lt$ and $Lt+2$.

  \noindent We take $L-1$ copies of $C_k^t$ and one copy of $C_k^{t+}$, leaving $t$ rows of weight $Lt+1$, and since $t$ is even, these rows can be paired and values swapped so that each row has weight $Lt$ or $Lt+2$. In particular, in columns $2,4,\ldots,t$, we swap the values in the pairs of rows $(1,2),(3,4),\ldots,(t-1,t)$, respectively.

  \noindent \textbf{Case 3}: weights $Lt+1$ and $Lt+2$.

  Let $r=\left\lceil\frac{k-t}{d-t}\right\rceil$ and let $a$ be the average column weight in the upper $d$ rows. To achieve average row weight at least $Lt+1$, $a$ must be such that the total weights in the upper and lower blocks are at least $d(Lt+1-r)$ and $(Lt+1)(k-d)$, respectively. This is guaranteed when both $a(Lk+1-r) \geq d(Lt+1-r)$ and $(t-a)(Lk+1-r) \geq (Lt+1)(k-d)$, or equivalently,
  
  \begin{equation}
    \label{eqn:first-even-d>t}
  \frac{d(Lt+1-r)}{Lk+1-r}\ \leq\ a\ \leq\ \frac{t(Ld+1-r)-k+d}{Lk+1-r}.
  \end{equation}

  To achieve average row weight at most $Lt+2$, $a$ must be such that the total weights in the upper and lower blocks are at most $d(Lt+2-r)$ and $(k-d)(Lt+2)$, respectively. This is guaranteed when both $a(Lk+1-r) \leq d(Lt+2-r)$ and $(t-a)(Lk+1-r) \leq (k-d)(Lt+2)$, or equivalently,  

  \begin{equation}
      \label{eqn:second-even-d>t}
      \frac{t(Ld+1-r)-2(k+d)}{Lk+1-r}\ \leq\ a\ \leq\ \frac{d(Lt+2-r)}{Lk+1-r}.
  \end{equation}

  Finally, to ensure that each block of rows is tall enough to accommodate the 1's specified in the construction, it suffices that $0 \leq a \leq d$ and $0 \leq t-a \leq k-d$, which together are equivalent to $\max(0,d+t-k) \leq a \leq \min(d,t)$. By our assumptions, this reduces to $0 \leq a \leq t$.

  We now show that these inequalities can all be simultaneously satisfied. In~\ref{eqn:first-even-d>t}, the upper bound is at least the lower bound if and only if $r \geq \frac{k-t}{d-t}$, which holds for our choice of $r$, and in~\ref{eqn:second-even-d>t}, the upper bound is at least the lower bound if and only if $r \leq  \frac{k-t}{d-t} +\frac{k}{d-t}$, which holds since
  $\frac{k}{d-t}\geq 1$. Exchanging the upper/lower bound pairs in~\ref{eqn:first-even-d>t} and~\ref{eqn:second-even-d>t} results in two pairs of inequalities on $a$ that are always satisfied. Finally, for $0 \leq a \leq t$, it suffices that at least one of the lower bounds is nonnegative, and at least one of the upper bounds is at most $t$. The lower bound in~\ref{eqn:first-even-d>t} is nonnegative if and only if $L \geq \frac{r-1}{t}$, and the upper bound is at most $t$ if and only if $L \geq -\frac{1}{t}$, so it suffices to take $L \geq \frac{r-1}{t}$. The rest of the proof follows the same reasoning as in Proposition~\ref{prop:d>t}. 
\end{proof}

\begin{proposition}\label{prop:d<t-even}
  Let $k \geq 3$, $1 < t < k$, and $1 \leq d < t$ be such that $t+d\leq k$ and $(k,t,d)\equiv (1,0,0) \pmod{2}$. Let $\bx$ be a tuple of weight $d$. Then $\Pol(\inn{t}{k}\cup
  \{\bx\},\NAE)$ does not have 2-block-symmetric functions of all odd arities.
\end{proposition}

\begin{proof}
  We place $L\geq 1$ copies of $C_k^t$ in the odd block of our tableau, leaving $Lk-1$ columns to be filled in the even block for a total arity of $2Lk-1$. The three weights used are $Lt$, $Lt-1$, and $Lt-2$, with $L$ determined later. 

  \noindent \textbf{Case 1}: weights $Lt$ and $Lt-1$.

  \noindent We take $L-1$ copies of $C_k^t$ and one copy of $C_k^{t-}$, so that each row has weight $Lt$ or $Lt-1$. 

  \noindent \textbf{Case 2}: weights $Lt$ and $Lt-2$.

  \noindent We take $L-1$ copies of $C_k^t$ and one copy of $C_k^{t-}$, leaving $t$ rows of weight $Lt-1$, and since $t$ is even, these rows can be paired and values swapped so that each row has weight $Lt$ or $Lt-2$. In particular, in columns $2,4,\ldots,t-2$, we swap the values in the pairs of rows $(1,2),(3,4),\ldots,(t-3,t-2)$, respectively. Finally, in column $k-1$, we swap the values in rows $k$ and $t-1$. 

  \noindent \textbf{Case 3}: weights $Lt-1$ and $Lt-2$.

  Let $r=\left\lceil\frac{k-t}{t-d}\right\rceil$ and let $a$ be the average column weight in the upper $d$ rows. To achieve average row weight at most $Lt-1$, $a$ must be such that the total weights in the upper and lower blocks are at most $d(Lt-1-r)$ and $(Lt-1)(k-d)$, respectively. This is guaranteed when both $a(Lk-1-r) \leq d(Lt-1-r)$ and $(t-a)(Lk-1-r) \leq (Lt-1)(k-d)$, or equivalently,
  
  \begin{equation}
    \label{eqn:first-even-d<t}
    \frac{t(Ld-1-r)+k-d}{Lk-1-r}\ \leq\ a\ \leq\ \frac{d(Lt-1-r)}{Lk-1-r}.
  \end{equation}

  To achieve average row weight at least $Lt-2$, $a$ must be such that the total weights in the upper and lower blocks are at least $d(Lt-2-r)$ and $(k-d)(Lt-2)$, respectively. This is guaranteed when both $a(Lk-1-r) \geq d(Lt-2-r)$ and $(t-a)(Lk-1-r) \geq (k-d)(Lt-2)$, or equivalently,  

  \begin{equation}
      \label{eqn:second-even-d<t}
      \frac{d(Lt-2-r)}{Lk-1-r}\ \leq\ a\ \leq\ \frac{t(Ld-1-r)+2(k-d)}{Lk-1-r}.
  \end{equation}

  Finally, to ensure that each block of rows is tall enough to accommodate the 1's specified in the construction, it suffices that $0 \leq a \leq d$ and $0 \leq t-a \leq k-d$, which together are equivalent to $\max(0,d+t-k) \leq a \leq \min(d,t)$. By our assumptions, this reduces to $0 \leq a \leq d$.

  We now show that these inequalities can all be simultaneously satisfied. In~\ref{eqn:first-even-d<t}, the upper bound is at least the lower bound if and only if $r \geq \frac{k-t}{t-d}$, which holds for our choice of $r$, and in~\ref{eqn:second-even-d<t}, the upper bound is at least the lower bound if and only if $r \leq  \frac{k-t}{t-d} +\frac{k}{t-d}$, which holds since
  $\frac{k}{t-d}\geq 1$. Exchanging the upper/lower bound pairs in~\ref{eqn:first-even-d<t} and~\ref{eqn:second-even-d<t} results in two pairs of inequalities on $a$ that are always satisfied. Finally, for $0 \leq a \leq d$, it suffices that at least one of the lower bounds is nonnegative, and at least one of the upper bounds is at most $d$. The lower bound in~\ref{eqn:second-even-d<t} is nonnegative if and only if $L \geq \frac{r+2}{t}$, and the upper bound in~\ref{eqn:first-even-d<t} is at most $d$ if and only if $t \leq k$, so it suffices to take $L \geq \frac{r+2}{t}$. The rest of the proof follows the same reasoning as in Proposition~\ref{prop:d>t}. 
\end{proof}

\section{Removing tuples}
\label{sec:rem}

In this section we prove Theorem~\ref{thm:t-in-k} and show how it implies
Theorem~\ref{thm:rem}.

Schaefer's dichotomy theorem (Theorem~\ref{thm:schaefer}) allows us to obtain a
simple description of all $\T$ with $\CSP(\T)$ tractable and $\inn{t}{k}\to\T$.

\begin{proposition}\label{prop:tink}
  Let $k \geq 3$, $1 \leq t < k,$ and suppose that $\inn{t}{k}\to \T$. Then $\CSP(\T)$ is tractable if and only if
\begin{enumerate}
  \item $0^k \in \T$ or $1^k \in \T$, or\label{const}
  \item $t$ is odd, $k$ is even, and $\T=\odd{k}$. \label{allodd}
\end{enumerate}
\end{proposition}

Observe that Proposition~\ref{prop:tink} in particular implies Proposition~\ref{prop:tinkhard}, NP-hardness of $\CSP(\inn{t}{k})$.

\begin{proof}
  In Case~\ref{const}, $\Pol(\T)$ contains a
  constant function so $\CSP(\T)$ is tractable by Theorem~\ref{thm:schaefer}, and in Case~\ref{allodd}, tractability is given by Proposition~\ref{proposition:oddinktract}. 

  We now turn to hardness. For the rest of the proof, assume that neither
  (\ref{const}) nor (\ref{allodd}) of the proposition statement applies.

  Suppose that $\inn{t}{k} \to \T$ by the function $\phi:\{0,1\}\to \{0,1\}.$ If
  $\phi$ is constant, then either $\T=\{0^k\}$ or $\T=\{1^k\}$
  and we are in case~(\ref{const}), a contradiction. If $\phi(x)=1-x$, note that
  $\phi(\inn{t}{k})$ satisfies conditions~(\ref{const}) and~(\ref{allodd})
  precisely when $\inn{t}{k}$ does, so it suffices to consider only the case
  where $\phi$ is the identity and $\inn{t}{k} \subseteq \T$.

  We show that $\Pol(\T)$ contains none of the functions $\AND_2$, $\OR_2$,
  $\MAJ_3$, and $\XOR_3$ from Theorem~\ref{thm:schaefer}. To accomplish this, we
  assume that one of these functions $f$ is present in $\Pol(\T)$, and then show
  that repeated application of $f$ to a certain set of tuples leads to
  (\ref{const}) or (\ref{allodd}) of the proposition, a contradiction. Recall
  that we denote by $f(R)$ the image of the relation $R$ under $f$.
  We make seven claims below, which together exclude the four functions as
  polymorphisms: case~(\ref{case:and}) covers $\AND_2$, case~(\ref{case:or})
  covers $\OR_2$, (overlapping) cases~(\ref{case:maj1}) and~(\ref{case:maj2})
  cover $\MAJ_3$, and cases~(\ref{case:xor1}), (\ref{case:xor2}),
  and~(\ref{case:xor3}) cover $\XOR_3$. Case~(\ref{case:xor3}) contradicts
  (\ref{allodd}) of the proposition; all others contradict
  (\ref{const}).\footnote{Another way of establishing this result for $\XOR_3$
  is via linear algebra (being closed under $\XOR_3$ is the same as being an
  affine subspace) and for $\MAJ_3$ from the fact that relations closed under
  $\MAJ_3$ are determined by their binary projections.}

  We give more details for case~(\ref{case:and}) for illustration. Taking $\AND_2$ of
  the two tuples, we get
  $\AND_2(1^t0^{k-t},01^t0^{k-t-1})=01^{t-1}0^{k-t}$, a tuple of weight $t-1$.
  By symmetry, we obtain all tuples of weight $t-1$, which is the statement of
  case~(\ref{case:and}). Continuing this way, we
  obtain all tuples of weight $t-2$, $t-3$, etc. until we eventually obtain $0^k$, which gives a contradiction as we assume that
  (\ref{const}) of the proposition does not apply, so $0^k\not\in\T$.

  \begin{enumerate}[i]
    \item \label{case:and} $\inn{(t-1)}{k} \subseteq \AND_2(\inn{t}{k})$

    tuples: $1^t0^{k-t}$ and $01^t0^{k-t-1}$ 

    eventual output: $0^k$

  \item \label{case:or} $\inn{(t+1)}{k} \subseteq \OR_2(\inn{t}{k})$ 

    tuples: $1^t0^{k-t}$ and $01^t0^{k-t-1}$ 

    eventual output: $1^k$

  \item \label{case:maj1} if $t \geq 2$ then $\inn{(t+1)}{k} \subseteq \MAJ_3(\inn{t}{k})$

    tuples: $1^{t-2}0^{k-t-1}110, 1^{t-2}0^{k-t-1}101,$ and $1^{t-2}0^{k-t-1}011$ 

    eventual output: $1^k$

  \item \label{case:maj2} if $t \leq k-2$ then $\inn{(t-1)}{k} \subseteq \MAJ_3(\inn{t}{k})$

    tuples: $1^{t-1}0^{k-t-2}100, 1^{t-1}0^{k-t-2}010,$ and $1^{t-1}0^{k-t-2}001$ 

    eventual output: $0^k$

  \item \label{case:xor1} if $t$ is even, then $\inn{(t-2)}{k} \subseteq \XOR_3(\inn{t}{k})$

    tuples: $1^{t-2}0^{k-t-1}110, 1^{t-2}0^{k-t-1}101,$ and $1^{t-2}0^{k-t-1}011$ 

    eventual output: $0^k$

  \item \label{case:xor2} if $t$ and $k$ are odd, then $\inn{(t+2)}{k} \subseteq \XOR_3(\inn{t}{k})$

    tuples: $1^{t-1}0^{k-t-2}100, 1^{t-1}0^{k-t-2}010,$ and $1^{t-1}0^{k-t-2}001$ 

    eventual output: $1^k$

  \item \label{case:xor3} if $t$ is odd and $k$ is even,\\
      then $\inn{(t+2)}{k} \subseteq \XOR_3(\inn{t}{k})$ if $t<k-2$, \\
      and $\inn{(t-2)}{k} \subseteq \XOR_3(\inn{t}{k})$ if $t>2$.

    tuples: 

    $1^{t-1}0^{k-t-2}100, 1^{t-1}0^{k-t-2}010,$ and $1^{t-1}0^{k-t-2}001$ 

    $1^{t-2}0^{k-t-1}110, 1^{t-2}0^{k-t-1}101,$ and $1^{t-2}0^{k-t-1}011$

    eventual output: all odd  weight tuples
\end{enumerate}
\end{proof}

With Proposition~\ref{prop:tink} in hand, we can prove Theorems~\ref{thm:t-in-k} and~\ref{thm:rem}.

\begin{theorem*}[Theorem~\ref{thm:t-in-k} restated]
  Let $k \geq 3$ and let $\T \subseteq \{0,1\}^k$ be a relation such that $\inn{t}{k} \to \T$ and 
  $\CSP(\T)$ is NP-hard. Then $\PCSP(\inn{t}{k},\T)$ is tractable if and only if
  $\T=\NAE$.
\end{theorem*}

\begin{proof}
  By Proposition~\ref{prop:firstsym}, we can assume that $\T$ is symmetric. If
  $\T=\NAE$ then $\PCSP(\inn{t}{k},\T)$ is tractable by
  Proposition~\ref{prop:naeeasy}. Otherwise, we show that $\Pol(\inn{t}{k},\T)$
  does not contain any of the tractable polymorphism families identified in the
  symmetric Boolean PCSP dichotomy (Theorem~\ref{thm:symmtrac}), and therefore $\PCSP(\inn{t}{k},\T)$
  is NP-hard.

  The families we need to rule out are constants, $\OR$, $\AND$, $\XOR$, $\AT,$ and $\THRo{q}$ for $q
  \in \mathbb{Q}$, as well as their negations. We deal first with the
  non-negated families. 
  Since $\CSP(\T)$ is NP-hard, by Proposition~\ref{prop:tink}, we have $0^k
  \not\in \T$ and $1^k \not\in \T$. Hence, $\Pol(\inn{t}{k},\T)$ does not contain constants. 

  Let $C_k^t$ be the $k \times
  k$ matrix containing the $k$ cyclic shifts of the column $1^t0^{k-t}$.
  Then $C_k^t$ prevents the polymorphism families $\OR$, $\AND$, $\XOR$ (if $k$
  is odd), and $\THRo{q}$ for all $q \neq \frac{t}{k}$. The case $q=\frac{t}{k}$
  is ruled out by~\cite[Fact B.3]{fkos19}, and $\AT$ is ruled out by~\cite[Claim~4.6]{BG21}.
  Since $\CSP(\T)$ is NP-hard, by Proposition~\ref{prop:tink},
  it remains to show that for even $k$, $\Pol(\inn{t}{k},\T)$ excludes $\XOR$ when
  $t$ is even, and likewise when $t$ is odd and $\T$ is missing a tuple of odd
  weight. 
    
  Let $k$ and $t$ be even. Applying $\XOR_k$ to the matrix
  $C^t_k$ returns the tuple $0^k$, so applying $\XOR_{k-1}$ to the first $k-1$
  columns of $C^t_k$ returns the last column $1^{t-1}0^{k-t}1$. We can ``fill
  in'' the 0's in the output by swapping 0/1 pairs of values in the input
  matrix. In particular, in the columns $k-1,k-3,\ldots,t+1$, we swap the
  entries in the pairs of rows $(k-1,k-2), (k-3,k-4),\ldots,(t+1,t)$,
  respectively. The resulting $k \times (k-1)$ matrix $M$ then satisfies
  $\XOR_{k-1}(M)=1^k$ and the arity $k-1$ is odd as required. An example with
  swapped values in bold is illustrated in Figure~(\ref{fig:xora}).

\begin{figure}[ht]
   \centering
  \begin{minipage}[b]{.4\linewidth}
  \centering
    \[\XOR_{7}\left(\begin{array}{ccccccc}
    1&0&0&0&0&1&1 \\ 
    1&1&0&0&0&0&1 \\
    1&1&1&0&0&0&0 \\
    1&1&1&1&\mathbf{1}&0&0 \\ 
    0&1&1&1&\mathbf{0}&0&0 \\
    0&0&1&1&1&1&\mathbf{1} \\
    0&0&0&1&1&1&\mathbf{0} \\ 
    0&0&0&0&1&1&1
  \end{array}\right)=
  \begin{array}{c}
    1\\1\\1\\1\\1\\1\\1\\1
  \end{array}
  \]
  \subcaption{$t=4$ and $k=8$.\label{fig:xora}}
  \end{minipage}
  \hfill
  \begin{minipage}[b]{.4\linewidth}
     \centering
    \[ 
\XOR_5\left(\begin{array}{ccccc}
    1&0&0&1&1\\ 
    1&1&0&0&1\\
    1&1&1&0&0\\
    0&1&1&1&0\\ 
    0&0&1&1&1\\ \hline
    0&0&0&0&0
  \end{array}\right)=
  \begin{array}{c}
    1\\1\\1\\1\\1\\0
  \end{array}
    \]
  \subcaption{$t=3$, $k=6$, and $d=5$.\label{fig:xorb}}
  \end{minipage}
  \caption{XOR.}
\end{figure}

Now let $k$ be even, $t$ be odd, and suppose that $\T$ does not contain the tuple
$\bx=1^d0^{k-d}$ of odd  weight $d$. By Observation~\ref{obs:invert}, we can assume without loss of generality that $t<d$. Then $\XOR_d$ applied to the
matrix $C_d^t$ padded with $k-d$ rows of 0's returns $\bx$. 
An illustration is given in Figure~(\ref{fig:xorb}). Therefore
$\XOR\not\subseteq\Pol(\inn{t}{k},\T)$.

\textbf{Negations}: Let $F$ be a family of functions. We reduce the task of
showing $\overline{F} \not\subseteq \Pol(\inn{t}{k},\T)$ to the already
completed task of showing $F\not\subseteq \Pol(\inn{t}{k},\T)$. Let $\bx \in
\{0,1\}^k \setminus \T$, let $f \in F$ be a function of arity $m$, and let $M$
be a $k \times m$ matrix of inputs to $f$ whose columns are $\inn{t}{k}$ tuples.
We established $F \not\subseteq \Pol(\inn{t}{k},\T)$ by finding $f$ and $M$ with
$f(M)=\bx$, and in the remaining cases we must find $\overline{f}\in
\overline{F}$ and $M$ such that $\overline{f}(M)=\bx$. But since
$\overline{f}(M)=\bx \Leftrightarrow f(M)=\overline{\bx}$, it suffices to find
$f \in F$ such that $f(M)=\overline{\bx}$, where
$\overline{\bx}=(1-x_1,\ldots,1-x_k)$ if $\bx=(x_1,\ldots,x_k)$.

The families $\overline{\AND}$, $\overline{\OR}$, $\overline{\XOR}$ (except
when $k$ is even and $t$ is odd), and
$\overline{\THRo{q}}$ for all $q \neq \frac{t}{k}$
are excluded from $\Pol(\inn{t}{k},\T)$ in the
same way as $\AND$, $\OR$, $\XOR$, and $\THRo{q}$ with the same matrices serving as
counterexamples. In detail, the matrix $C_k^t$, which contains $k$ cyclic shifts
of the column $1^t0^{k-t}$, prevents the polymorphism families
$\overline{\AND}$, $\overline{\OR}$, $\overline{\XOR}$ (if $k$ is odd), and
$\overline{\THRo{q}}$ (if $q\neq\frac{t}{k}$). The
case $\overline{\XOR}$ with $k$ even, $t$ even is ruled out the same way as before,
illustrated in Figure~(\ref{fig:xora}).

To see that $\overline{\AT}$ and $\overline{\THR{t}{k}}$ (with $q=\frac{t}{k}$) are
also excluded, let $\bx \not\in \T$ be a tuple of weight $d\neq t$. Then the tuple $\overline{\bx}$ of  weight $k-d$ can be
returned by an $\AT$ function~\cite[Claim~4.6]{BG21} and a $\THR{t}{k}$ function~\cite[Fact B.3]{fkos19}. If $k-d=t$, then the $\AT$ and $\THR{t}{k}$ functions of arity 1 output
$\overline{\bx}$ on input $\overline{\bx}$.

Finally, when $k$ is even, $t$ is odd, and $\T$ does not contain the tuple
$\bx$ of odd  weight $d$, the $\XOR$ argument above (illustrated in
Figure~\ref{fig:xorb}) applies since
$\overline{\bx}$ also has odd  weight $k-d$. Again, if $k-d=t$, then the $\XOR$
function of arity 1 outputs $\overline{\bx}$ on input $\overline{\bx}$.
\end{proof}

\begin{theorem*}[Theorem~\ref{thm:rem} restated]
  Let $k \geq 3$ and $\emptyset\neq S \subseteq (\inn{t}{k})^c \cap \NAE$.
  If $t$ is odd, $k$ is even, and $S$ contains tuples of only even 
  weight, then $\PCSP(\inn{t}{k},\NAE\setminus\Ss)$ is tractable. Otherwise,
  $\PCSP(\inn{t}{k},\NAE\setminus\Ss)$ is NP-hard.
\end{theorem*}

\begin{proof}
  The tractability in the first statement of the theorem is proved in
  Section~\ref{sec:results}. Otherwise, $t$ is even, or $k$ is odd, or $S$
  contains a tuple of odd  weight. Take $\T=\NAE\setminus\Ss$. Observe that case~(\ref{const})
  of Proposition~\ref{prop:tink} does not apply as neither $0^k$
  nor $1^k$ is part of the template. Moreover, case~(\ref{allodd}) of
  Proposition~\ref{prop:tink} does not apply either: If $t$ is odd and $k$ is even
  then $S$ contains a tuple of odd  weight and hence $\NAE\setminus\Ss$
  cannot have all odd  weight tuples. Thus, by
  Proposition~\ref{prop:tink}, $\CSP(\T)$ is NP-hard. Then, by
  Theorem~\ref{thm:t-in-k},
  $\PCSP(\inn{t}{k},\T)=\PCSP(\inn{t}{k},\NAE\setminus\Ss)$ is NP-hard.
\end{proof} 

\section*{Acknowledgements}

We would like to thank the anonymous referees of both
the conference~\cite{bz21:icalp} and this full version of the paper. We also
thank Kristina Asimi and Libor Barto for useful discussions regarding the
content of this and their paper~\cite{Asimi21:mfcs}.

{\small
\bibliographystyle{plainurl}
\bibliography{bz22}
}

\end{document}